\documentclass[graybox]{svmult}

\usepackage{type1cm}        
\usepackage{makeidx}         
\usepackage{graphicx}        
\usepackage{multicol}        
\usepackage[bottom]{footmisc}

\usepackage{kbordermatrix}

\usepackage{newtxtext}       %
\usepackage[varvw]{newtxmath}       

\usepackage{amsmath,amsfonts,amsthm}
\usepackage{enumerate}
\usepackage{xcolor}

\usepackage[normalem]{ulem}

\def\sig{\varsigma}
\def\tr{{\operatorname{tr}}}
\def\re{{\operatorname{re}}}
\def\im{{\operatorname{im}}}
\def\mtx{{\operatorname{M}}}
\def\QM{{\operatorname{QM}}}
\def\ux{{\underline x}}
\def\uX{{\underline X}}
\def\state{\mathscr{S}}
\def\cS{\mathcal{S}}
\def\C{\mathbb{C}}
\def\R{\mathbb{R}}
\def\N{\mathbb{N}}
\def\cB{\mathcal{B}}
\def\cD{\mathcal{D}}
\def\cH{\mathcal{H}}
\newcommand{\Langle}{\mathop{<}\!}
\newcommand{\Rangle}{\!\mathop{>}}
\def\mx{\Langle \ux\Rangle}
\def\px{\C\Langle \ux\Rangle}
\def\ncstate{\state\!\mx}
\def\sa{{\rm sa}}
\def\fin{{\rm fin}}

\DeclareMathOperator{\wt}{wt}

\newcommand{\ket}[1]{\mathinner{|#1\rangle}}

\newcommand{\dyad}[1]{| #1\rangle \langle #1|}
\newcommand{\ot}[0]{\otimes}
\newcommand{\one}[0]{I}
\newcommand{\av}[1]{{\langle#1\rangle}}
\newcommand{\ax}[1]{{\langle x\rangle}}
\newcommand{\ay}[1]{{\langle y\rangle}}

\newcommand{\SEP}[0]{\mathrm{SEP}}

\newcommand{\WW}[0]{\mathcal{W}}

\makeindex             

\begin{document}

\title*{Positivity of state, trace, and moment polynomials, and applications in quantum information}
\titlerunning{State, trace, and moment polynomials}
\author{Felix Huber\orcidID{0000-0002-3856-4018} and\\ Victor Magron\orcidID{0000-0003-1147-3738} and \\Jurij Vol\v{c}i\v{c}\orcidID{0000-0002-7848-3458}}
\institute{Felix Huber 
\at 
Institute of Theoretical Physics and Astrophysics, Faculty of Mathematics, Physics and Informatics, University of Gdańsk,
Wita Stwosza 57, 80-308 Gdańsk, Poland, \email{felix.huber@ug.edu.pl}
\and Victor Magron \at LAAS CNRS, 7 avenue du Colonel Roche,
F-31031 Toulouse, France, \email{vmagron@laas.fr}
\and Jurij Vol\v{c}i\v{c} \at Department of Mathematics,  
Drexel University, Pennsylvania, \email{jurij.volcic@drexel.edu}
}
%
%
\maketitle


\abstract{
state, trace, and moment polynomials are polynomial expressions in several operator or random variables and positive functionals on their products (states, traces or expectations).
While these concepts, and in particular their positivity and optimization, arose from problems in quantum information theory, yet they naturally fit under the umbrella of multivariate operator theory.
This survey presents state, trace, and moment polynomials in a concise and unified way, and highlights their similarities and differences. The focal point is their positivity and optimization. Sums of squares certificates for unconstrained and constrained positivity (Positivstellens\"atze) are given, and parallels with their commutative and freely noncommutative analogs are discussed. They are used to design a convergent hierarchy of semidefinite programs for optimization of state, trace, and moment polynomials. 
Finally, circling back to the original motivation behind the derived theory, multiple applications in quantum information theory are outlined.
}

\section{Introduction}
\label{sec:intro}
Polynomial optimization problems appear in many fields of science and engineering. 
They consist of minimizing a given polynomial $f$ in several commutative variables $\ux= (x_1,\dots,x_n)$ on a set defined by finitely many polynomial (in)equalities. 
Since \emph{exact} computation of the minimum of $f$ turns out to be NP-hard \cite{laurent2009sums}, an intense amount of research rather focused on computing \emph{approximations}. 
This optimization problem is strongly related to the problem of representing nonnegative polynomials with \emph{sums of squares} (SOS) certificates, these latter objects being at the cornerstone of real algebraic geometry \cite{marshall2008positive}. 
Artin's solution of Hilbert's 17th problem allows one to represent a globally nonnegative polynomial $f$ as an SOS of rational functions, while in the constrained case, Schm\"udgen's Positivstellensatz \cite{schm93} and Putinar's Positivstellensatz \cite{putinar1993positive} certify the positivity of $f$ with weighted SOS polynomials. 
The duality between polynomials positive on a set defined by polynomial (in)equalities and probability measures supported on this set connects SOS-based representations with necessary conditions to solve the moment problem on this set; see the monographs \cite{lasserre2009moments,schmbook} dedicated to this topic. 
About two decades ago, Lasserre provided in \cite{lasserre2001global} a general hierarchical framework to approximate the minimum of $f$ from below. 
Under mild assumptions, slightly stronger than compactness of the feasible set, this latter hierarchy converges to the minimum of $f$ as a consequence of Putinar's Positivstellensatz \cite{putinar1993positive}. 
At each step of the hierarchy, a lower bound $\alpha$ of the minimum is obtained by means of semidefinite programming (SDP) \cite{vandenberghe1996semidefinite}, a special class of  convex optimization problems consisting of minimizing a linear function under linear matrix inequality constraints. 
In addition to $\alpha$, the SDP program returns a weighted SOS decomposition of $f - \alpha$, each weight corresponding to a polynomial involved in the set of constraints. Overall this weighted SOS decomposition certifies the nonnegativity of $f -\alpha$ on the feasible set. 
Similar techniques can analyze dynamical systems involving polynomial input data. 
We refer the interested reader to the recent monograph  \cite{henrion2020moment} focusing on control and analysis of such dynamical systems. 

The above concepts of positivity and SOS decompositions can be extended to the noncommutative setting by replacing the commutative variables $\ux = (x_1,\dots,x_n)$ by freely noncommutative variables. 
Multivariate polynomial inequalities in such matrix or operator variables appear in many scientific fields, including quantum statistical mechanics \cite{carlen2010trace, sutter2017multivariate}, 
control systems \cite{sig,frag}, quantum information theory \cite{beigi2013sandwiched, pozas2019bounding}, operator algebras \cite{netzer2010tracial,fritz2014can} and free probability \cite{guionnet2009free,junge2013}.
The associated problems are typically dimension-independent: that is, one is interested in validity of inequalities for matrices of all sizes, or one wishes to optimize an objective function over bounded operators on Hilbert spaces of arbitrary dimension.
This dimension-free aspect impedes direct applicability of results from real algebraic geometry to these multivariate operator inequalities. However, the fundamental principles nevertheless transfer into the noncommutative setting, sometimes with even stronger implications than in the commutative setting. A seminal example is the resolution of the noncommutative analog of Hilbert's 17th problem, established independently by Helton \cite{hel} and McCullough \cite{mcc}: noncommutative polynomials that are positive semidefinite on all matrix tuples are sums of squares of noncommutative polynomials.
This breakthrough naturally led to new possibilities in noncommutative polynomial optimization.
Under an analogous assumption as in the commutative setting, converging SDP-based hierarchies allow one to approximate as closely as desired the minimal eigenvalue or trace of a noncommutative polynomial subject to finitely many polynomial inequality constraints, as a consequence of Helton-McCullough Positivstellensatz \cite{helton2004positivstellensatz}. 
The dual approach via the noncommutative moment problem was developed independently, motivated by quantum information theory.
The groundbreaking Navascu\'es-Pironio-Ac\'in hierarchy   \cite{npa08} allows one to readily compute upper bounds over the maximal violation levels of linear Bell inequalities (see also \cite{doherty2008quantum}), while \cite{pna10} generalizes this dual SDP-hierarchy to general noncommutative polynomial optimization problems. 
Furthermore, while the above developments concern eigenvalue optimization of noncommutative polynomials, similar advances on trace optimization of noncommutative polynomials were also achieved \cite{burgdorf13}. Its applications include
bounds on entanglement dimensions \cite{grib18},
matrix factorization ranks \cite{grib19}, and perfect strategies for synchronous quantum non-local games \cite{benewatts,brannan}.
The related algorithmic developments have been implemented in various libraries, such as 
\texttt{Ncpol2sdpa} \cite{wittek2015algorithm}, 
\texttt{NCSOStools} \cite{cafuta2011ncsostools}, 
\texttt{NCalgebra} \cite{helton1996ncalgebra}, 
\texttt{TSSOS} \cite[Appendix~B]{magron2023sparse},    available in Python, Matlab, Mathematica and Julia, respectively. 
We refer the interested reader to \cite{bhardwaj2022noncommutative} for a short introduction on noncommutative polynomial optimization and to \cite{burgdorf2016optimization} for a more comprehensive treatise. 

Recent advances in quantum information theory have motivated  new developments in the field of operator theory to  
characterize positivity and optimize wider classes of polynomial functions on operator tuples: \emph{state}, \emph{trace}, and \emph{moment} polynomials. 
These objects are polynomials in operator variables and either states, traces, or expectations of their products, respectively. 
This survey summarizes the related research efforts and outlines applications in entanglement detection, quantum network correlations, quantum uncertainties, and quantum codes.

\section{Noncommutative and state polynomials}\label{sec:state}

Throughout the text let $\ux=(x_1,\dots,x_n)$ be freely noncommuting variables. Let $\mx$ be the set of all words in $\ux$ (the free monoid generated by $\ux$), with the empty word 1. The free $*$-algebra $\px$ of \emph{noncommutative polynomials} in $\ux$ with coefficients $\C$ is endowed with a unique $\C$-skew linear involution $*$ that makes $x_j$ hermitian, $x_j^*=x_j$.\footnote{
Likewise, one can set up framework for non-hermitian variables $z_j$, and consider the free $*$-algebra $\C\langle z_1,\dots,z_n,z_1^*,\dots,z_n^*\rangle$ generated by them and their formal adjoints. All the ensuing statements correspond to analogs for $\C\langle x_1,\dots,x_{2n}\rangle$ generated by $2n$ hermitian variables. Thus, we restrict to hermitian variables for simplicity.
} 
To the monoid $\mx$ we also assign a commutative algebra as follows. With each $1\neq w\in\mx$ we associate a commutative indeterminate $\sig(w)$, and let 
$$\state=\C\big[\sig(w)\colon w\in\mx\setminus\{1\}\big].$$
Then $\state$ is a polynomial $*$-algebra on countably many generators, with a $\C$-skew linear involution determined by $\sig(w)^*=\sig(w^*)$. Elements of $\state$ are called \emph{state polynomials}. Combining commutative and noncommutative aspects, the $*$-algebra of \emph{noncommutative state polynomials} is $\ncstate=\state\otimes_{\C} \px$. The $\sig$ notation extends to a natural $\state$-linear unital map $\sig:\ncstate\to\state$.
The real subspaces of self-adjoint elements $f=f^*$ in $\px,\,\state,\,\ncstate$ are denoted 
$\px^\sa,\,\state^\sa,\,\ncstate^\sa$, and we write $\re (g)=\frac12(g+g^*)$, $\im (g)=\frac{1}{2i}(g-g^*)$.

Noncommutative state polynomials are therefore formal polynomial expressions involving the noncommuting variables and $\sig$ symbols in their words. In this survey, we wish to furthermore view them as multivariate functions in operator variables and states. To do this, we first recall some standard notions from functional analysis \cite{Reed80,Tak02}.
Let $\cH$ be a complex Hilbert space. Let $\cB(\cH)$ denote the $*$-algebra of bounded operators on $\cH$. When $\cH=\C^k$ we identify $\cB(\cH)$ with $k\times k$ matrices $M_k(\C)$.
If $A\in\cB(\cH)$ then $A\succeq0$ denotes that $A$ is a positive semidefinite operator.
A \emph{state} on $\cH$ is a positive unital linear functional $\lambda:\cB(\cH)\to\C$. Every unit vector $u\in\cH$ determines a state $\lambda(X)=\langle Xu,u\rangle$, 
which is called a \emph{vector} state. 
More generally, if $\rho\in\cB(\cH)$ is a density operator (a positive semidefinite trace-class operator with trace 1), then $\lambda(X)=\tr(\rho X)$ is a state.
Let $\cS(\cH)$ denote the set of all states on $\cB(\cH)$.

Let $f\in\ncstate$ be a noncommutative state polynomial, $\uX=(X_1,\dots,X_n)$ a tuple of self-adjoint bounded operators on a Hilbert space $\cH$, and $\lambda:\cB(\cH)\to\C$ a state. The \emph{state evaluation} of $f$ at $(\lambda,\uX)$ is $f(\lambda,\uX)\in\cB(\cH)$ defined in a natural way, by replacing $x_j$ and $\sig$ in $f$ with $X_j$ and $\lambda$, respectively. In the same way we define state evaluations for state polynomials; if $f\in\state$ then $f(\lambda,\uX)\in\C$. The involution on noncommutative state polynomials is compatible with the usual involution on $\cB(\cH)$; in particular, if $f\in\ncstate^\sa$ (resp. $f\in\state^\sa$) then $f(\lambda,\uX)$ is self-adjoint (resp. real) for all $\lambda$ and $\uX$ as above.

One can generalize this formalism by adding several independent state symbols $\sig_1,\sig_2,\dots$ and consider evaluations on tuples of operators and several states, and then suitably extend the subsequent results presented in this survey; however, for the sake of simplicity we persist in working with a single formal state symbol $\sig$.

\subsection{State semialgebraic sets}

A central topic of this paper is positivity of state polynomials. More precisely, we are interested in certifying whether a state polynomial is nonnegative under all state evaluations at operators and states that satisfy given noncommutative state constraints. To formalize this, let $C\subset\ncstate^\sa$ be a set of constraints.
Let $\cH$ be a separable infinite-dimensional Hilbert space; note that up to isomorphism, there is only one such Hilbert space, so one may fix $\cH=\ell^2(\N)$.
To $C$ we assign two closely related \emph{state semialgebraic sets}:
\begin{align*}
\cD^\infty(C)&= \big\{(\lambda,\uX)\in\cS(\cH)\times \cB(\cH)^n\colon 
X_j^*=X_j, \ c(\lambda,\uX)\succeq0 \text{ for all } c\in C\big\},\\
\cD^\fin(C)&= \bigcup_{k\in\N}\big\{(\lambda,\uX)\in\cS(\C^k)\times \mtx_k(\C)^n\colon X_j^*=X_j, \ c(\lambda,\uX)\succeq0 \text{ for all } c\in C\big\}.
\end{align*}
The terminology originates from real algebraic geometry \cite{Mar}, where semialgebraic sets are solution sets of polynomial inequalities.
Let us briefly compare the two variants of state semialgebraic sets. As mentioned above, $\cD^\infty(C)$ is essentially independent of the concrete realization of a separable infinite-dimensional Hilbert space $\cH$. Also, one can view 
$\cD^\fin(C)$ as a subset of $\cD^\infty(C)$ (by embedding the matrix algebras $M_k(\C)$ into $\cB(\cH)$ as unital $*$-subalgebras). 
However, $\cD^\fin(C)$ may be much smaller than $\cD^\infty(C)$ (from the perspective of state polynomial positivity), as the following example shows.

\begin{example}\label{e:findim}
Write $z=x_1+ix_2$ and let $C=\{z^*z-1\}$.
Note that $(X_1,X_2)\in \cD^\fin(C)$ precisely when the singular values of $Z=X_1+i X_2$ lie in $[1,\infty)$. Let $f=\sig(zz^*)-1$. Then $f\ge0$ on $\cD^\fin(C)$ because $ZZ^*-I\succeq 0$ for any square matrix $Z$ satisfying $Z^*Z\succeq I$. On the other hand, $f\not\ge0$ on $\cD^\infty(C)$. Indeed, let $R\in\cB(\ell^2(\N))$ be the right shift, and $X_1=\re(R)$, $X_2=\im(R)$. Note that $R^*R=I$.
Let $\lambda$ be the vector state given by the vector $u=(1,0,0,\dots)\in\ell^2(\N)$; notice that $R^*u=0$.
Then $(\lambda,\uX)\in \cD^\infty(C)$ because $R^*R=I$, and
$$f(\lambda,\uX)
=\langle RR^*u,u\rangle-1=-1.$$
\end{example}

Example \ref{e:findim} shows that state polynomial positivity on all finite-dimensional Hilbert spaces is in general strictly weaker than state polynomial positivity on infinite dimensional Hilbert spaces.
There are natural conditions under which there is no gap between finite-dimensional and infinite-dimensional picture, e.g., when $\cD^\fin(C)$ is matrix convex \cite{cvxpos11}, which is particularly of interest in control theory. 
On the other hand, in quantum information theory, one of the most important sources for noncommutative optimization applications, no-communication imposes commutation relations on pairs of operator variables; in this case, the gap between $\cD^\fin(C)$ and $\cD^\infty(C)$ exists (by the negative answer to Connes' embedding problem \cite{connes}) and carries fundamental implications for the quantum theory.

From the perspective of state polynomial positivity, one can restrict $\cD^\infty(C)$ and $\cD^\fin(C)$ to only vector states, and all the results in this survey stay unchanged. This is due to the Gelfand-Naimark-Segal (GNS) construction, which allows one to replace general states with vector states (cf. purification of mixed quantum states in quantum physics).

As defined above, state semialgebraic sets are solution sets of noncommutative state polynomial inequalities. This framework in particular encloses equality constraints. For example, if one wants to restrict operators $\uX$ satisfying a noncommutative polynomial equation $p(\uX)=0$, then one can achieve that by considering the constraint set containing $\pm \re(p),\pm \im(p)$.

\subsection{State quadratic modules}\label{sec:qm}

Let $C\subset\ncstate^\sa$ be a constraint set. Next we look at state polynomials which are nonnegative on $\cD^\infty(C)$ ``for obvious reasons''.
The \emph{state quadratic module} $\QM(C)$ generated by $C$ is the subset of $\state$ consisting of elements of the form
\begin{equation}\label{e:qm}
\sum_j\sig(f_j^*c_jf_j),\qquad f_j\in \ncstate,\ c_j\in\{1\}\cup C.
\end{equation}
It is straightforward to see that elements \eqref{e:qm} are nonnegative on $\cD^\infty(C)$, and so $q\ge0$ on $\cD^\infty(C)$ for all $q\in\QM(C)$. Section \ref{sec:pos} below investigates weak converses of this observation.

From the perspective of real algebraic geometry \cite{Mar}, $\QM(C)$ is a quadratic module in $\state$, in the sense that it is closed under addition and multiplication by squares of elements from $\state$. 
Following this terminology, $\QM(C)$ is an archimedean quadratic module in $\state$ if for every $f\in\state^{\sa}$ there exists $N\in\N$ such that $N-f,N+f\in\QM(C)$.
Archimedeanity is an important property in real algebraic geometry and operator algebras, and plays a crucial role in various commutative and noncommutative Positivstellens\"atze on bounded domains \cite{Put,HM}. In the state polynomial framework, we will say that the set of constraints $C$ is \emph{archimedean} if there exists $N>0$ such that
\begin{equation}\label{e:arch}
N-x_1^2-\cdots-x_n^2 = \sum_j f_j^*c_jf_j \quad \text{ for some }f_j\in\ncstate,\ c_j\in \{1\}\cup C.
\end{equation}
If $C$ is archimedean, then $\QM(C)$ is archimedean as a quadratic module in $\state$.
Furthermore, note that if $C$ is archimedean, then $\cD^\infty(C)$ is bounded. Indeed, if $N>0$ is such that \eqref{e:arch} holds, then $\|X_j\|\le \sqrt{N}$ for all $\uX\in \cD^\infty(C)$ and $j=1,\dots,n$. Conversely, if $\cD^\infty(C)$ is bounded in the sense that $\|X_j\|\le N$ for all $\uX\in\cD^\infty(C)$, one can simply add $nN^2-\sum_j x_j^2$ to $C$ to obtain an archimedean constraint set.

Archimedeanity is a mild condition in applications of state polynomial positivity. Many optimization problems pertaining to state polynomials are concerned with bounded feasible domains.
For example, state polynomial optimization problems in quantum physics are typically interested in operator variables that are projections, and thus include constraints of the form $\pm(x_j-x_j^2)$ for $j=1,\dots,n$. Such constraint sets are archimedean because
$$n-\sum_{j=1}^n x_j^2=\sum_{j=1}^n(1-x_j)^2+\sum_{j=1}^n 2(1-x_j^2),$$
which has the desired form \eqref{e:arch}.

\section{Trace and moment polynomials}\label{ss:trmom}

There are two other constructions closely related to state polynomials, namely trace polynomials \cite{ksv17,huber21,tropt20} and moment polynomials \cite{blekherman2022,mompop}.

Trace polynomials originate in invariant theory, and their operator evaluations emerge in free probability and von Neumann algebras.
As the name indicates, trace polynomials are polynomial expressions of traces of products of variables, subject to the fundamental property of the trace $\tr(ab)=\tr(ba)$. There are two natural ways of evaluating such trace symbols: either as a usual trace (e.g. on matrices, or trace-class operators) or as a normalized trace $\tr I=1$ (e.g. rescaled trace on matrices, or tracial states on C*-algebras). The first option is well-understood when one is interested in evaluations on matrices of fixed dimension. On the other hand, the second option can be viewed as a special case of state evaluations where one restricts to tracial states. Namely, one considers state polynomials on state semialgebraic sets $\cD^\infty(C)$ where $C$ contains
$$\pm \re(\sig(uv)-\sig(vu)),\quad \pm \im(\sig(uv)-\sig(vu))\qquad 
\text{for all }u,v\in\mx.$$

It is important to stress that this dimension-independent framework applies only to evaluations involving normalized trace. Positivity of trace evaluations on matrices of all sizes using the (usual) non-normalized trace behaves quite differently. For example, deciding whether a univariate polynomial is nonnegative on all matrix evaluations with the normalized trace is a semialgebraic problem, and in particular such trace polynomials admit SOS representations with denominators \cite{univar20}; in contrast, deciding whether a univariate polynomial is nonnegative on all matrix evaluations with the non-normalized trace cannot reduce to a semialgebraic problem, leading to certain undecidability phenomena \cite{vandermonde}.

The other relatives of state polynomials are moment polynomials, which naturally arise in probability, statistics, and moment problems. Moment polynomials are polynomial expressions in mixed moments of random variables on a probability space. By the spectral theorem (see e.g. \cite{schmUnbded}), states on polynomials in $n$ commuting bounded self-adjoint operators correspond to integrals of polynomials in random variables with respect to Borel probability measures. Thus moment polynomials may be viewed as state polynomials in commuting variables. The commutativity condition can be encoded into constraints, so moment polynomials are realized as state polynomials on state semialgebraic sets $\cD^\infty(C)$ where $C$ contains
$$\pm \re(x_ix_j-x_jx_i),\quad \pm \im(x_ix_j-x_jx_i)\qquad 
\text{for all }1\le i<j\le n.$$

Due to the above relationships, some results for state polynomial positivity on state semialgebraic with reasonably general sets of constraints (e.g. Theorem \ref{t:arch} below) carry direct implications for trace and moment polynomials. 
In practice, one usually implements the additional relations directly into the framework by appropriating the defining features of the state symbol ($\sig(uv)=\sig(vu)$ for trace polynomials, $x_ix_j=x_jx_i$ for moment polynomials).
While problems on trace or moment polynomial positivity can be reduced to problems on state polynomial positivity as explained above, this perspective is sometimes too reductive. 
Namely, tracial and commutative constraints are rather special in a certain sense. In order to obtain compelling results about trace and moment polynomial positivity, it is important to leverage deeper specific results associated with these constraints. 
Concretely, one can often apply either invariant and representation theory of the symmetric group \cite{pro,huber21, huber2023refutingspectralcompatibilityquantum} and commutative algebra and moment problems \cite{schmbook,blekherman2022} to obtain stronger statements for trace and moment polynomials, respectively.

\subsection*{A comparative example}\label{exa:sttrmom}

The following toy example (in two operator variables $x_1=x$ and $x_2=y$) illustrates the distinction between state, trace, and moment polynomials (for a more elaborate example arising from a nonlinear Bell inequality, see \cite[Example 7.2.3]{stateopt23}). 
It also demonstrates that establishing dimension-independent operator inequalities may be intricate even in rather simple-looking instances (below we demonstrate that for the trace polynomial framework).

\noindent {\bf State regime}.
We look at the values of the state polynomial $$f=\frac12\Big(\sig(xyxy)+\sig(yxyx)\Big)-\sig(xyx)\sig(y)$$
subject to constraints $x^2=x$ and $y^2=y$ (i.e., we are interested in its state evaluations on pairs of projections).
Notice that $f(\lambda,(X,Y))=-\frac{1}{16}$ for
$$X=\begin{pmatrix}
1&0\\0&0
\end{pmatrix},\quad
Y=\frac12\begin{pmatrix}
1&1\\1&1
\end{pmatrix},\qquad
\lambda(\cdot)=\langle \,\cdot\, v,v\rangle
\quad\text{for}\quad
v=\frac12\begin{pmatrix}
\sqrt{2-\sqrt{2}} \\
-\sqrt{2+\sqrt{2}}
\end{pmatrix}.
$$

\noindent {\bf Trace regime.} In contrast, as a trace polynomial (i.e., a state polynomial restricted to tracial states), the minimum of $f$ is $-\frac{1}{27}$. Indeed, for every pair of projections $X,Y$ and a tracial state $\tau$ one has 
$$f(\tau,(X,Y))=\tau(XYXY)-\tau(XY)\tau(Y).$$
It is well-known (see e.g. \cite{pedersen})
that, up to unitary change of basis, every irreducible pair of projections is one of the following:
$$(0,0),\quad (1,0),\quad (0,1),\quad (1,1),\quad \Bigg(\begin{pmatrix}1&0\\0&0\end{pmatrix}, 
\begin{pmatrix}t&\sqrt{t-t^2}\\ \sqrt{t-t^2}& 1-t\end{pmatrix}\Bigg) \text{ for }0<t<1.$$
Let us denote these pairs as $(X_1,Y_1),\dots,(X_4,Y_4),(X_5(t),Y_5(t))$.
Thus, there exists a probability measure $\pi$ on $\{0,1\}^2\cup ]0,1{\mkern-2mu}[$ such that $(X,Y)$ is unitarily equivalent to the direct integral of the above irreducible representations with respect to $\pi$, while $\tau$ is the direct integral of the corresponding normalized traces (see \cite[Section IV.8]{Tak02} for a comprehensive study of direct integrals). 
We can decompose $\pi$ as
$$\pi=\alpha_1\delta_{(0,0)}+\alpha_2\delta_{(1,0)}
+\alpha_3\delta_{(0,1)}+\alpha_4\delta_{(1,1)}+\alpha_5\nu,$$
where $\delta_{(i,j)}$ is the Dirac delta measure at $(i,j)$, $\nu$ is a probability measure on the interval $]0,1{\mkern-2mu}[$, and $\alpha_j\ge0$, $\sum_j\alpha_j=1$. Then
$$
\tau\big(p(X,Y)\big)
=\sum_{j=1}^4\alpha_j p(X_j,Y_j)
+\alpha_5\int \frac12\tr p\big(X_5(t),Y_5(t)\big)\,{\rm d} \nu(t),
$$
for every bivariate noncommutative polynomial $p$, where $\frac12\tr$ is the normalized trace on $2\times 2$ matrices. In particular,
\begin{align*}
f(\tau,(X,Y))
&=\left(\alpha_4+\alpha_5\int \frac{t^2}{2}\,{\rm d} \nu \right)
-\left(\alpha_4+\alpha_5\int \frac{t}{2}\,{\rm d} \nu \right)\left(\alpha_3+\alpha_4+\alpha_5\int \frac12\,{\rm d} \nu \right)\\
&=\alpha_4\left(1-\left(\alpha_3+\alpha_4+\frac{\alpha_5}{2}\right)\right)+\frac{\alpha_5}{2}
\int \left(t^2-\left(\alpha_3+\alpha_4+\frac{\alpha_5}{2}\right)t\right)\,{\rm d} \nu\\
&\ge \alpha_4\left(1-\left(\alpha_3+\alpha_4+\frac{\alpha_5}{2}\right)\right)-\frac{\alpha_5}{8}\left(\alpha_3+\alpha_4+\frac{\alpha_5}{2}\right)^2\\
&=:p(\alpha_3,\alpha_4,\alpha_5),
\end{align*}
which is at least $-\frac{1}{27}$ by standard constrained optimization arguments (namely, by evaluating $p$ at its stationary points in the simplex 
$\{\alpha_3,\alpha_4,\alpha_5\ge0,\,
\alpha_3+\alpha_4+\alpha_5\le1\}$, and at stationary points of $p$'s restriction to the boundary).
The minimum $f(\tau,(X,Y))=-\frac{1}{27}$ is attained at
$$X=\begin{pmatrix}
1&0&0\\0&0&0\\0&0&0
\end{pmatrix},\quad
Y=\frac13\begin{pmatrix}
1&\sqrt{2}&0\\\sqrt{2}&2&0\\0&0&3
\end{pmatrix},\qquad
\tau \text{ the normalized trace on }3\times3\text{ matrices.}
$$

\noindent {\bf Moment regime}.
Finally, as a moment polynomial (i.e., a state polynomial restricted to commuting operators), $f$ can be interpreted as a polynomial expression in mixed moments with respect to a probability measure $\mu$ of two commuting random variables $X,Y$ valued in $\{0,1\}$. 
In this case,
\begin{align*}
f(\mu,(X,Y))
&=\frac12\left(\int(XY)^2\,{\rm d}\mu+\int(YX)^2\,{\rm d}\mu\right)-\int XYX\,{\rm d}\mu\int Y\,{\rm d}\mu \\
&=\int XY\,{\rm d}\mu-\int XY\,{\rm d}\mu\int Y\,{\rm d}\mu\\
&=\int XY\,{\rm d}\mu\int (1-Y)\,{\rm d}\mu\ge0.
\end{align*}
An SOS approach (assisted with semidefinite optimization) for establishing lower bounds on state and tracial evaluations of $f$ is outlined in Example \ref{exa:sttrmom2} below.

\section{Two Positivstellens\"atze}\label{sec:pos}

This section presents two main SOS certificates for positive state polynomials. They are essentially disjoint in scope, assumptions, implications and proof techniques.

\subsection{Hilbert-Artin theorem for state polynomials}

First we consider state polynomial positivity in the completely unconstrained case.
The following theorem can be seen as the resolution of a state polynomial variant of Hilbert's 17th problem.

\begin{theorem}[\cite{stateopt23}]\label{t:h17}
The following are equivalent for $f\in\state^{\sa}$:
\begin{enumerate}[(i)]
\item $f\ge0$ on $\cD^\fin(\emptyset)$;
\item $f\ge0$ on $\cD^\infty(\emptyset)$;
\item $f$ is a quotient of sums of products of elements of the form $\sig(h^*h)$ for $h\in\ncstate$.
\end{enumerate}
\end{theorem}

For example, $\sig(x_1^2)\sig(x_2^2)-\sig(x_1x_2)\sig(x_2x_1)$ is an everywhere nonnegative state polynomial (cf. Cauchy-Schwarz inequality); this can be certified in a purely algebraic way as
$$
\sig(x_1^2)\sig(x_2^2)-\sig(x_1x_2)\sig(x_2x_1)
=\frac{\sig\Big(
\big(\sig(x_2^2)x_1-\sig(x_2x_1)x_2\big)^*
\big(\sig(x_2^2)x_1-\sig(x_2x_1)x_2\big)
\Big)}{\sig(x_2^2)}.
$$

The implication (ii)$\Rightarrow$(i) is immediate, and the implication (iii)$\Rightarrow$(ii) is routine (special care is only required at state evaluations where the denominator of $f$ vanishes). On the other hand, the implication (i)$\Rightarrow$(iii) is the core of Theorem \ref{t:h17}, and it relies on results intertwining
real algebra and invariant theory \cite{ksv17,univar20} and a truncated algebraic variant of the GNS construction \cite{cvxpos11}.
One particular consequence of Theorem \ref{t:h17} is that for unconstrained state polynomial positivity,
there is no difference between operator state evaluations and matricial state evaluations of all finite dimensions (as opposed to the constrained setup, such as Example \ref{e:findim}).

The algebraic certificate (iii) in Theorem \ref{t:h17} has similar shortcomings than its classical analog (solution of Hilbert's 17th problem). Namely, quotients are indeed required; in other words, while not every nonnegative state polynomial are in $\QM(\emptyset)$, closing this set under products and quotients describes all nonnegative state polynomials. Furthermore, the complexity of the algebraic certificate in (iii) cannot be bounded by the degree of $f$ in general.

Trace and moment analogs of Theorem \ref{t:h17} fail.
More precisely, the tracial analog of Theorem \ref{t:h17} holds for $n=1$ (only one operator variable) \cite{univar20}. 
For (at least some) $n>1$, the failure of (i)$\Rightarrow$(ii) in Theorem \ref{t:h17} for trace polynomial positivity is equivalent to the refutation of Connes' embedding problem \cite{connes,trpos}.
Moreover, (ii)$\Rightarrow$(iii) in Theorem \ref{t:h17}
for trace polynomial positivity likewise fails; nevertheless, there is a weak SOS certificate with denominators for trace-positive noncommutative polynomials on tracial von Neumann algebras \cite{trpos}. 
Similarly,
moment polynomials nonnegative on all measures do not necessarily admit an SOS certificate as in Theorem \ref{t:h17}; concrete examples are given in \cite{blekherman2022,mompop}. Nevertheless, every nonnegative moment polynomial admits an SOS representation after an arbitrarily small perturbation of its (possibly high-degree) coefficients \cite{mompop}; this perturbative Positivstellensatz is established with intrinsically commutative tools and might not hold for state or trace polynomials.

\subsection{Archimedean Positivstellensatz for state polynomials}

Next, we consider state polynomials that are positive on bounded state semialgebraic sets. Let $C\subset\ncstate^\sa$ be a constraint set. Since we are interested in bounded domains, it is reasonable to assume that $C$ is archimedean, as seen in Section \ref{sec:state}. In general, there is a difference between positivity on $\cD^\fin(C)$ and positivity on $\cD^\infty(C)$ (one can modify Example \ref{e:findim} by enlarging $C$ to obtain an archimedean constraint set). The latter one admits an algebraic certificate in the spirit of Putinar's Positivstellensatz \cite{Put} for classical polynomials and the Helton-McCullough Positivstellensatz \cite{HM} for noncommutative polynomials.

\begin{theorem}[\cite{stateopt23}]\label{t:arch}
Let $C\subset\ncstate^\sa$ be an archimedean constraint set. The following are equivalent for $f\in\state^{\sa}$:
\begin{enumerate}[(i)]
\item $f\ge0$ on $\cD^\infty(C)$;
\item $f+\varepsilon\in\QM(C)$ for every $\varepsilon>0$.
\end{enumerate}
\end{theorem}

The implication (ii)$\Rightarrow$(i) is straightforward.
The proof of (i)$\Rightarrow$(ii) essentially decomposes into two steps, a commutative one and a noncommutative one. If $f$ is a state polynomial for which (ii) fails, then one first shows that there is a formal evaluation (homomorphism into $\C$) on $\state$ that is negative at $f$ but nonnegative on $\QM(C)$. This is a commutative problem, resolved by the Kadison-Dubois representation theorem that relies on $\QM(C)$ being archimedean. Afterwards, one uses the GNS construction (on the free algebra \cite{HM}) to show that such an evaluation is actually a state evaluation, resulting in a point in $\cD^\infty(C)$ where $f$ is negative.

Theorem \ref{t:arch} has a somewhat more analytic flavor than Theorem \ref{t:h17}. Also, while the positivity certificate in Theorem \ref{t:h17} requires denominators (and products), the one in Theorem \ref{t:arch} instead relies on the perturbation of the constant term. As explained in Section \ref{sec:conv} below, this feature is important for state polynomial optimization.

One can also derive a certificate for positivity of noncommutative state polynomials on bounded state semialgebraic sets involving an auxiliary self-adjoint noncommuting variable $x_0$ (freely independent of $x_1,\dots,x_n$).

\begin{corollary}\label{c:ncarch}
Let $C\subset\ncstate^\sa$ be an archimedean constraint set. 
Let $x_0$ be an auxiliary self-adjoint variable.
The following are equivalent for $f\in\ncstate^{\sa}$:
\begin{enumerate}[(i)]
\item $f\succeq0$ on $\cD^\infty(C)$;
\item $\sig(x_0fx_0)+\varepsilon\in\QM(C\cup\{\pm(1-x_0^2)\})$ for every $\varepsilon>0$.
\end{enumerate}
Here, $\QM(C\cup\{\pm(1-x_0^2)\})$ is a state quadratic module within state polynomials in variables $x_0,\dots,x_n$.
\end{corollary}

\begin{proof}
Note that $C'=C\cup\{\pm(1-x_0^2)\}$ is archimedean. Thus by Theorem \ref{t:arch} it suffices to show that (i) is equivalent to
\begin{enumerate}[(i')]
\item $\sig(x_0fx_0)\ge0$ on $\cD^\infty(C')$.
\end{enumerate}
First suppose that (i) holds. Let $(\lambda,X_0,\uX)\in \cD^\infty(C')$ be arbitrary. Then $f(\lambda,\uX)\succeq0$ by assumption, so 
$X_0f(\lambda,\uX)X_0\succeq0$ and thus
$$\sig(x_0fx_0)\big(\lambda,X_0,\uX\big)=
\lambda\big(X_0f(\lambda,\uX)X_0\big)\ge0,$$
so (i') holds.
Now suppose (i) does not hold, and let $(\lambda,\uX)\in\cD^\infty(C)$ be such that $f(\lambda,\uX)$ is not positive semidefinite. 
By the GNS construction, we can without loss of generality assume that $\lambda$ is a vector state, given by the unit vector $u\in\cH$. Let $v\in\cH$ be a unit vector such that $\langle f(\lambda,\uX)v,v\rangle<0$. It is easy to see that there exists a self-adjoint unitary $X_0$ such that $X_0u=\alpha v$ for some $\alpha\in\C$ with $|\alpha|=1$. Then
$(\lambda,X_0,\uX)\in \cD^\infty(C')$ and
$$\sig(x_0fx_0)\big(\lambda,X_0,\uX\big)
=\lambda\big(X_0f(\lambda,\uX)X_0\big)
=\langle X_0f(\lambda,\uX)X_0u,u\rangle
=\langle f(\lambda,\uX)v,v\rangle<0,
$$
so (i') does not hold.
\end{proof}

\section{SDP hierarchies for optimization of state, trace, and moment polynomials}\label{sec:conv}

Theorem \ref{t:arch} leads to a semidefinite programming (SDP) hierarchy for state polynomial optimization on bounded state semialgebraic sets.
Given $C\subset\ncstate$ and $f\in\state^\sa$ consider the optimization problem
\begin{equation}\label{e:optprob}
\alpha_*=\inf\big\{f(\lambda,\uX)\colon (\lambda,\uX)\in\cD^\infty(C)\big\}.
\end{equation}
Various problems in quantum physics can be reduced to state polynomial optimization problems of the form \eqref{e:optprob}; see Section \ref{sec:appl} and references within for examples.
For $d\in \N$ let
$$\QM(C)_d=\left\{
\sum_j\sig(f_j^*c_jf_j)\colon f_j\in \ncstate,\ 
c_j\in\{1\}\cup C,\ 
\deg(f_j^*c_jf_j)\le 2 d
\right\},$$
where $\deg$ denotes the total degree of a noncommutative state polynomial with respect to $x_1,\dots,x_n$ (appearing freely or within the state symbol $\sig$). Then the $\QM(C)_d$ form an increasing sequence of convex cones in $\state$, whose union equals $\QM(C)$. Denote
\begin{equation}\label{e:optprob1}
\alpha_d=\sup\big\{\alpha\in\R\colon f-\alpha \in \QM(C)_d\big\}.
\end{equation}
Computing $\alpha_d$ can be reformulated as an SDP.
Furthermore, Theorem \ref{t:arch} then implies the following.

\begin{corollary}\label{c:hier1}
Let $C\subset\ncstate^\sa$ be an archimedean constraint set, and $f\in\state^{\sa}$.
Then $(\alpha_d)_d$ is an increasing sequence converging to the solution $\alpha_*$ of \eqref{e:optprob}.
\end{corollary}

The SDP hierarchy \eqref{e:optprob1} and Corollary \ref{c:hier1} are also applicable to optimization of trace and moment polynomials by adding tracial and commutation relations, respectively, to $C$ as in Subsection \ref{ss:trmom}. In practice, one implements these relations (and other equality constraints from $C$) directly into symbolic manipulation while preparing the SDP input (as reduction rules, or using Gr\"obner bases).

\begin{example}\label{exa:sttrmom2}
Recall
$f=\frac12(\sig(xyxy)+\sig(yxyx))-\sig(xyx)\sig(y)$ on pairs of projections in the comparative example from Section \ref{exa:sttrmom}. Denote $C=\{\pm(x-x^2),\pm(y-y^2)\}$. Using an SDP solver one obtains
$$\sup\big\{\alpha\in\R\colon f-\alpha \in \QM(C)_5\big\}=-\frac{1}{16}$$
numerically. Together with Section \ref{exa:sttrmom}, this indicates that the minimum of $f$ on $\cD^\infty(C)$ is $-\frac{1}{16}$.
When only tracial states are considered (namely, $f$ is viewed as a trace polynomial), one extends $C$ to $C'=C\cup\{ \pm(\sig(uv-vu)\colon u,v\in\Langle x,y\Rangle\}$ and applies an SDP solver to compute
$$\sup\big\{\alpha\in\R\colon f-\alpha \in \QM(C')_4\big\}=-\frac{1}{27}.$$
This gives an alternative to the argument in Section \ref{exa:sttrmom} for the lower bound of the trace polynomial $f$ on pairs of projections.
Note that in this instance, proving that the state polynomial $f+\frac{1}{16}$ is nonnegative on pair of projections requires a sum of squares of degree $2\cdot 5=10$ (one can check that a lower degree is not sufficient); on the other hand, for the trace polynomial $f+\frac{1}{27}$ one only requires a sum of squares of degree $2\cdot 4=8$.
\end{example}

A comprehensive study of this SDP hierarchy is presented in \cite[Section 6]{stateopt23}, while the details about their refinements for trace polynomials and moment polynomials are given in \cite[Section 5]{tropt20} and \cite[Section 5]{mompop}, respectively.

For the purposes of this survey let us briefly address the dual SDP hierarchy. The dual of SDP \eqref{e:optprob1} corresponds to optimizing $L(f)$ over all unital linear functionals $L$ that are nonnegative on the convex cone $\QM(C)_d$. 
While solutions of an SDP and its dual SDP are generally distinct (that is, only weak duality is guaranteed), there is no duality gap for SDPs \eqref{e:optprob1} if the constraint set $C$ is archimedean in the following strong sense: there is $N>0$ such that
$$N-x_1^2-\cdots-x_n^2=\sum_j p_j^*p_j+\sum_k \lambda_kc_k,$$
for some affine $p_j\in\px$, $\lambda_k\in\R_{\ge0}$ and quadratic $c_k\in C$.
In particular, one can always satisfy this condition by simply adding $N-x_1^2-\cdots-x_n^2$ to the set.
Thus one can replace \eqref{e:optprob1} by its dual, which is sometimes more convenient for implementation. Furthermore, there is a sufficient condition that enables one to extract a finite-dimensional optimizer for \eqref{e:optprob1} out of its dual (this condition is based on the matrix rank of the semidefinite constraint of the dual solution).
Also, strong duality allows one to use primal-dual interior point methods, which are the most widely implemented SDP algorithms, for solving \eqref{e:optprob1}.

In practice, the dual SDP is implemented via (multivariate) Hankel matrices. 
As for quadratic modules, let $\ncstate^{\sa}_{d}$ be the subset of state polynomials with total degree at most $d$. 
For state polynomial optimization, we associate to a unital linear functional $L : \ncstate^{\sa}_{2d} \to \C$ the Hankel matrix $\Gamma(L)$ indexed by all words $u,v$ of $\ncstate^{\sa}_d$ with $\Gamma_{u,v}(L) = L(\sig(u^\star v))$. Then the positivity of $L$ relates to the positive semidefiniteness of its associated Hankel matrix $\Gamma(L)$. 
Similarly, to each $c \in C$ with $d_c=\lceil \deg c /2 \rceil$, we associate the so-called localizing matrix $\Gamma(c\, L)$ indexed by all words $u,v$ of $\ncstate^{\sa}_{d-d_c}$ with $\Gamma_{u,v}(c\,L) = L(\sig(u^\star c v))$. 
While the hierarchy \eqref{e:optprob1} is the analog of the Helton-McCullough approach to noncommutative polynomial positivity \cite{HM}, its dual is the analog of the renowned NPA hierarchy \cite{pna10}, and $\Gamma(L)$ plays the role of the noncommutative moment matrix in the NPA hierarchy.

As in the primal SDP hierarchy \eqref{e:optprob1}, one tends to implement equality constraints as reduction rules when preparing the SDP input. For the dual problem, there is also an alternative approach, proposed in \cite[Definition 3.2]{mor24}: once a collection of constraints satisfies the Archimedean condition, any further equality constraint $c=0$ can be imposed on $L$ by requiring $L(c^2)=0$ (instead of $\Gamma(\pm c\, L)\succeq0$).

As an example, the following are Hankel matrices for state, trace, and moment polynomial optimization problems (in two variables $x_1=x$ and $x_2=y$) at the second level of the hierarchy.

\noindent {\bf State regime}.
Suppose one wants to maximize  
$f = \sig(x)^2 + \sig(y)^2$ on $x,y$ such that $1- x^2 \succeq 0$ and $1 - y^2 \succeq 0$. 
Any Hankel matrix at the second level of the hierarchy is then indexed by all state monomials of degree at most two,
\begin{align*}
&1, x, y, \sig(x), \sig(y), 
x^2,y^2, xy, yx, x\sig(x), x\sig(y), y\sig(x), y\sig(y), \\
& \sig(x^2), \sig(y)^2, \sig(xy), \sig(yx),
\sig(x)^2, \sig(y)^2,
\sig(x)\sig(y).
\end{align*}
Given a unital linear functional $L : \ncstate^{\sa}_4 \to \C$, the associated Hankel matrix $\Gamma(L)$ is 
{\footnotesize
\begin{equation*}
  \kbordermatrix{
  & 1 & x & y & \sig(x) & \dots & xy& \dots& \sig(x)\sig(y)
  \\
  & 
  1 & L(\sig(x)) & L(\sig(y)) & L(\sig(x)) & & L(\sig(xy))& & L(\sig(x)\sig(y))
  \\
  & 
  & L(\sig(x)^2) & L(\sig(x)\sig(y)) & L(\sig(x)^2) & &L(\sig(x^2y))&  & L(\sig(x)^2\sig(y))
  \\
  & 
  & & L(\sig(y)^2) & L(\sig(y)\sig(x)) & &L(\sig(yxy))&& L(\sig(x)\sig(y)^2)
  \\
  & 
  & & & 
  L(\sig(x)^2) & & L(\sig(x)\sig(xy))&  & L(\sig(x)^2\sig(y))\\
  &
  &&&& 
  \\
  & 
  & & & & & L(\sig(yx^2y)) & 
  & L(\sig(yx) \sig(x)\sig(y))\\
  &
  &&&&&&
  \\
  &  
  &  & & & &  && L(\sig(x)^2\sig(y)^2)
  }.
\end{equation*}}
We can define similarly the two localizing matrices $\Gamma((1-x^2)L)$ and $\Gamma((1-y^2)L)$, indexed by all state monomials of degree at most one, i.e., $1, x, y, \sig(x), \sig(y)$. 
To find an upper bound on $f$, one then maximizes $L(f)$ over all unital linear functionals $L$ being nonnegative on $QM(C)_2$, or equivalently over the set of positive semidefinite matrices that have the same structure as the above-mentioned state Hankel and localizing matrices. 
The statement of Corollary \ref{c:hier1} then implies that when the indexing sequence includes all monomials of degree $d$ with $d \to \infty$, the optimum of $f$ is achieved. 

\noindent {\bf Trace regime}.
For trace polynomials we demand that the state is tracial,
satisfying $\re( \sig(uv) - \sig(vu)) = \im( \sig(uv) - \sig(vu)) = 0$ for all $u,v \in \Langle x,y \Rangle$.
In our example, this implies the 
$\sig(xy) = \sig(yx)$, 
$\sig(x^2y^2) = \sig(y^2x^2)$ , $\sig(xyxy) = \sig(yxyx)$, etc. 
on the set of tracial Hankel matrices.

\noindent {\bf Moment regime}.
For trace polynomials we demand that the operators commute, 
that is $uv = vu$ for all $u,v \in \Langle x,y \Rangle$.
This implies the constraints
$\sig(xy) = \sig(xy)$, 
$\sig(xyxy) = \sig(x^2y^2) = \sig(y^2x^2) = \sig(yxyx)$ etc. 
on the set of moment Hankel matrices.

It is thus clear that constraints become more restrictive from state to trace to moment regime. Or in other words, one has the inclusion $S_m \subseteq S_t \subseteq S_s$ where $S_m, S_t, S_s$ are the feasible sets involving moment, tracial and state Hankel matrices, respectively.

\section{Applications}\label{sec:appl}

Several problems in quantum information theory can be formulated as state, trace or moment polynomial optimization problems. In this section we present a selection of such problems, and sketch how the SDP hierarchy from Section \ref{sec:conv} applies to them.
For the sake of compatibility with the motivation literature, we adopt some of the physics notation in this section, such as ${}^\dag$ in place of ${}^*$, for example.

\subsection{Entanglement in Werner states}

Quantum entanglement accounts for the non-classical correlations appearing in multipartite quantum systems.
A quantum state $\varrho$ on $(\C^d)^{\ot n}$ (here we view $\varrho$ as a density operator) is \emph{separable}
if it can be written as a convex combination of product states,
\begin{equation*}
  \varrho = \sum_i p_i \varrho_{1}^{(i)} \ot \cdots \ot \varrho_{n}^{(i)},
\end{equation*}
where $p_i \geq 0$, $\sum_i p_i = 1$, and $\varrho_j^{(i)}$ are states on $\C^d$.
If no such decomposition exists, the state is termed \emph{entangled}.
Determining whether a state is separable or entangled is a computationally hard problem~\cite{GURVITS2004448, DBLP:journals/qic/Gharibian10}.

Let us denote the set of separable states in $(\C^d)^{\ot n}$ as $\mathrm{SEP}(n,d)$ and the set of entangled states as $\mathrm{ENT}(n,d)$.
A method to detect entanglement is that of entanglement witnesses:
a hermitian operator is an entanglement witness if 
\begin{align*}
    \tr(\mathcal{W} \varrho) &\geq 0
    \quad \text{for all} \quad \varrho \in \mathrm{SEP}(n,d), \quad\text{and}\\
    \tr(\mathcal{W} \varphi) &<0
\quad \text{for some} \quad \varphi \in \mathrm{ENT}(n,d).
\end{align*}
Thus a witness acts as a separating hyperplane in the state space for the set of separable states.
Here we show how witnesses for the so-called Werner states
can be obtained from the trace polynomial optimization framework. What follows is based on \cite{Huber_2022} (see also \cite{Oberwolfach2021_2140b, huber21}).

Werner states are invariant under the adjoint diagonal action of unitaries,
that is, $U^{\ot n} \varrho (U^{\ot n})^\dag = \varrho$ for all $U \in \mathcal{U}(d)$.
The Schur-Weyl duality states that the commutant of $U^{\ot n}$ on $(\C^d)^{\ot n}$
is the algebra generated by the permutation operators
$\eta_d(\sigma)$ for $\sigma \in S_n$, whose action is given by
\begin{equation*}
 \eta_d(\pi) \ket{v_1} \ot  \cdots \ot \ket{v_n} = \ket{v_{\sigma^{-1}(1)}} \ot \cdots \ot \ket{v_{\sigma^{-1}(n)}}.
\end{equation*}
As a consequence, one can write
$
 \varrho = \sum_{\sigma \in S_n} a_\sigma \eta_d(\sigma)
$ with $a_\sigma \in \C$.
The same invariance can be imposed upon $\WW$, and so one writes $\WW = \sum_{\sigma \in S_n} w_\sigma \eta_d(\sigma)$.
Furthermore, under this symmetry the Hilbert space decomposes as
\begin{equation*}
 (\C^d)^{\ot n} \simeq \bigoplus_{\substack{\lambda \vdash n \\ \operatorname{ht}(\lambda) \leq d}} \mathcal{U}_\lambda \ot \mathcal{S_\lambda},
\end{equation*}
where
the unitary group $\mathcal{U}(d)$ acts on the spaces $\mathcal{U}_\lambda$ and the symmetric group $S_n$ on the spaces $\mathcal{S_\lambda}$,
and the representations are labeled by partitions $\lambda$ of $n$ of height at most $d$.
As a consequence, the state $\varrho$ block-diagonalizes as
\begin{equation*}
 \varrho =
 \bigoplus_{\substack{\lambda \vdash n \\ \operatorname{ht}(\lambda) \leq d}}
 \tilde \one \ot \varrho_\lambda ,
\end{equation*}
where $\tilde \one $ is the maximally mixed state on $\mathcal{U}_\lambda$ (i.e., a suitable multiple of the identity).
This allows to show the following.

\begin{theorem}[\cite{Huber_2022}]\label{thm:dim_free}
 For every entangled Werner state there exists a witness $w\in \C S_n$ satisfying
\begin{equation*}
  \tr\big(\eta_d(w) \varrho\big) \geq 0 \quad \text{for all } \varrho \in \SEP(d,n),\, \forall d \in \N .
\end{equation*}
\end{theorem}
Such witness has a nonnegative expectation value on all separable states in all dimensions,
and is thus called \emph{dimension-free}.
The key idea in the proof is that given a witness $w$ in dimension $d'<d$,
it can be promoted to a witness $\tilde w = w + u$ in all dimensions by adding
an element $u\succeq 0$ whose support lies exclusively in partitions of height larger than $d$ (but at most $n$).
The existence of such $u$ was established in \cite{Huber_2022}.

Now consider the task of finding such witnesses.
If a permutation $\sigma$ decomposed into cycles as
$\sigma = (\alpha_1 \dots \alpha_r) \dots (\xi_1 \dots \xi_t) \in S_n$,
then the following identity holds:
\begin{equation*}
 \tr\big( \eta_d(\sigma^{-1})  (X_1 \ot \cdots \ot X_n) \big)
 =
 \tr\big(X_{\alpha_1}  \cdots X_{\alpha_r}\big) \cdots
 \tr\big(X_{\xi_1}  \cdots X_{\xi_t}\big).
\end{equation*}
That is, permutations acting on tensor products are mapped to trace polynomials.
For an element $w \in \C S_n$ to correspond to a dimension-free Werner state witness,
one requires that $\min \tr\big( \eta_n(w) \varrho\big) \geq 0$ for all $\varrho \in \SEP(n,n)$.
Such minimum is attained at the extreme points of the separable set,
which consists of pure product states.
One thus requires that
\begin{equation*}
p = \min_{\ket{\phi_1}, \dots, \ket{\phi_n} \in \C^n} \tr\big(\eta_n(w) \dyad{\phi_1} \ot \cdots \ot \dyad{\phi_n}\big) \geq 0,
\end{equation*}
which corresponds to a trace polynomial in the variables 
$\dyad{\phi_1}, \dots, \dyad{\phi_n}$. 
We now strengthen the domain of nonnegativity in the above optimization problem
by asking for the positivity of $p$ for all projections from any tracial von Neumann algebra.
This can be done through the SDP hierarchy as in Section \ref{sec:conv} (see \cite{Huber_2022} for details).

\subsection{Nonlinear Bell inequalities}

One of the motivating applications for state polynomial optimization is finding maximal violations of nonlinear Bell inequalities, as outlined in this section (see \cite{stateopt23,ligthart22} for details). 

A quantum network \cite{Fri,pozas2019bounding,tavakoli22} consists of several non-communicating parties with measurements (modelled by projection-valued measures), some of which share independent multipartite quantum states. The correlation of such a model refers to the array of conditional expectations of joint measurement outputs relative to given inputs. 
The correlation is classical if it can be reproduced by a model with shared classical randomness in place of quantum states.
The simplest (and most well-understood) quantum networks are multipartite Bell scenarios, where several parties share a single common state. The set of possible correlations is convex in this setup, and its geometry is studied by linear Bell inequalities (functionals that are nonnegative on the set of correlations).
While verifying validity of a Bell inequality on the set of classical correlations is straightforward (because of convexity, one can restrict to checking finitely many deterministic models), to establish Bell inequalities on the set of quantum correlations one typically utilizes the NPA hierarchy \cite{npa08} for noncommutative polynomial optimization.

However, the set of correlations for a network that is not a multipartite Bell scenario is in general not convex (not even if restricted to classical correlations). For this reason, there is no hope of describing it using linear Bell inequalities; a reasonable next step is to study the set of correlations using nonlinear Bell inequalities (yet there is no guarantee that these completely distinguish between different networks; to resolve this, the inflation technique was recently introduced \cite{inflation}). 
Nonlinear Bell inequalities are given polynomial expressions in conditional expectations arising from the model. When restricted to classical correlations, these are moment polynomials in probability measures that govern shared randomness in the network. 
When all quantum correlations are considered, these are state polynomials whose noncommuting variables correspond to the projective measurements, and the formal state corresponds to the quantum state of the model. 

Let us add more details about the above outline for the bilocal scenario, which is the simplest network with a nonconvex set of correlations. 
The network consists of three non-communicating parties and an independent pair of states, where the middle party shares a state with each of the other parties. This scenario is represented by three Hilbert spaces $\cH_A,\cH_{B'}\otimes\cH_{B''},\cH_C$, ensembles of projections $A_i\in\cB(\cH_A), B_j\in\cB(\cH_{B'}\otimes \cH_{B''}), C_k\in\cB(\cH_C)$, and two states $\rho_{AB}$ on $\cH_A\otimes \cH_{B'}$ and $\rho_{BC}$ on $\cH_{B''}\otimes \cH_C$. The correlation of this model is given by
\begin{equation}\label{e:biloc}
\tr\Big(
\big(\rho_{AB}\otimes \rho_{BC}\big)
\big(
A_i\otimes B_j\otimes C_K
\big)
\Big),
\end{equation}
where with a slight abuse of notation the tensors are arranged in a compatible fashion.
The idea of modeling \eqref{e:biloc} with commutative symbols representing states on operators, and mimicking the tensor product structure with factorization relations on these symbols, originated in \cite{pozas2019bounding} and \cite{mor24}, where \emph{scalar extension} and \emph{factorization} hierarchies are designed to recognize correlations of the bilocal network. 
In the language of this survey, polynomial Bell inequalities for classical correlations of the bilocal network are then given by moment polynomials $p=p(\sig(a_ib_jc_k)\colon i,j,k)$ that are nonnegative subject to relations
\begin{alignat*}{2}
&a_i^2=a_i,\ b_j^2=b_j,\ c_j^2=c_k, &&\text{\small(projections)}\\
&a_i,b_j,c_k \text{ commute,} &&\text{\small(classical \& no communication)}\\
&\sig(ac)=\sig(a)\sig(c) \text{ for all products }a,c \text{ of }a_i,c_k. \quad &&\text{\small(bilocality)}
\end{alignat*}
On the other hand, examples of polynomial Bell inequalities for quantum correlations are given by state polynomials $p=p(\sig(a_ib_jc_k)\colon i,j,k)$ that are nonnegative subject to relations
\begin{alignat*}{2}
&a_i^2=a_i,\ b_j^2=b_j,\ c_j^2=c_k,
&&\text{\small(projections)}\\
&a_ib_j=b_ja_i,\ b_jc_k=c_kb_j,\ c_ka_i=a_ic_k, &&\text{\small(no communication)}\\
&\sig(ac)=\sig(a)\sig(c) \text{ for all products }a,c \text{ of }a_i,c_k.
 \quad &&\text{\small(bilocality)}
\end{alignat*}
By the refutation of Connes' embedding problem \cite{connes}, there exist (linear) Bell inequalities for quantum models as defined above that do not correspond to such nonnegative state polynomials. However, the latter correspond to polynomial Bell inequalities for correlations of suitably generalized quantum models (see \cite{ligthart22} for a comparison of various considered models).
By translating polynomial Bell inequalities in the language of moment and state polynomials, one can then use the convergent SDP hierarchy from Section \ref{sec:conv} to establish their validity.

There are (at least) two other SDP hierarchies that are tailored to Bell inequalities arising from quantum networks.
The so-called polarization and inflation hierarchies \cite{ligthart21,ligthart22,mor24} are based on techniques arising from quantum information theory.
A fundamental tool for proving their convergence is a quantum version of de Finetti theorem in quantum probability \cite{raggio89,caves01,brandao12,ligthart21}.
Let us very roughly compare the SOS hierarchy of Section \ref{sec:conv} with the polarization/inflation hierarchy. To handle operator variables, all hierarchies rely on the GNS construction. The differences arise from the treatment of pure state constraints.
Whereas the SOS hierarchy views state symbols as commuting indeterminates, the polarization and inflation hierarchies explore the symmetry in pure state constraints, and use polarization to replace them by linear constraints with permutation invariance. The relation between these two approaches is  analogous to the identification of the quotient of a tensor product space modulo commutation relations on one hand with the subspace of symmetric tensors on the other hand.
While the SOS hierarchy then relies on the Kadison-Dubois representation theorem for positivity of commutative polynomials, the convergence of the polarization/inflation hierarchy is derived from the quantum de Finetti theorem.

For the bilocal network, as well as for star networks, the relationship between various quantum models of correlations is now well-understood \cite{ligthart22,mor24}. However, distinguishability for networks beyond this family remains an open problem.

\subsection{Quantum uncertainty relations}
\label{sec:uncert}

Uncertainty among noncommuting observables is a fundamental feature of quantum mechanics. The most well-known is perhaps the Heisenberg uncertainty relation, stating that the standard deviations of the position and the momentum of a particle cannot both be known with arbitrary precision,
\begin{equation*}
\Delta x \cdot \Delta p  \geq \frac{\hbar}{2\pi}  .
\end{equation*}
Here $\Delta^2 A = \langle A^2\rangle - \langle A\rangle^2$
is the variance of an observable $A$,
and $\langle A\rangle = \tr(\varrho A)$ is its expectation value
with respect to some state $\varrho$.
Similar relations hold for observables in a finite Hilbert space. The question then is: how can one derive uncertainty relations - also among more than two observables - in a systematic manner?

An interesting scenario is that of a set of unitary and hermitian observables
$\{A_i\}_{i=1}^n$, $A_i^\dag = A_i$, $A_i A_i^\dag=\one$,
that mutually either commute or anti-commute,
$A_i A_j = (-1)^{\chi_{ij}} A_j A_i$ with $\chi_{ij} \in \{0,1\}$.
Define the quantity
\begin{equation*}
 \beta = \operatorname{sup}_{\varrho, A_i, \mathcal{H}} \quad \sum_{i=1}^n  \langle A_i\rangle^2,
\end{equation*}
where the optimization is over
all states $\varrho$ and observables $A_i$
on Hilbert spaces $\mathcal{H}$ that support the defining relations of the $A_i$.
Then a tight additive uncertainty relation is as follows:
\begin{equation*}
 \sum_{i=1}^n \Delta A_i^2 \geq n - \beta.
\end{equation*}
Our aim is now to provide efficiently computable upper bounds on $\beta$.

Let $\vartheta(G)$ be the Lov\'asz theta number of a graph~$G$~\cite{GALLI2017159}~\footnote{
Several equivalent formulations are known,
see \cite{Knuth1994, porumbel2022demystifyingcharacterizationsdpmatrices}.},
\begin{equation*}
 \vartheta(G) = \Big\{ \max\, \tr(X)
 \,\big|\, \begin{pmatrix}
            1 &x^T \\
			x &X
           \end{pmatrix}  \succeq 0,\,
           x_i = X_{ii} \,\forall i,\, X_{ij} = 0 \,\text{ if }\, i \sim j \Big\} .
\end{equation*}
It is known that $\vartheta$ upper bounds
the independence number $\alpha$ of a graph,
providing an SDP approximation
to a quantity that is NP-hard to compute.
The relations among the observables can be encoded into an \emph{anti-commutativity graph} $G$ defined by
$i \sim j$ if $A_i A_j = - A_j A_i$ and $i\not \sim j$ otherwise.
In \cite{de2022uncertainty} it is shown that $\beta < \vartheta(G)$,
implying the additive uncertainty relation~\footnote{
A related bound
$\langle \sum_{i=1}^n a_i A_i \rangle^2 \leq \vartheta(G)$ for all $|\!|a|\!|_2 = 1$ in the context of optimizing fermionic Hamiltonians is given in \cite{hastings2022optimizing}.
}
\begin{equation*}
    \sum_{i=1}^n \Delta^2 A_i \geq n - \vartheta(G).
\end{equation*}

Let us show how this bound naturally fits into a complete SDP hierarchy
originating from the state polynomial optimization framework~\cite{Mor_n_2024}.
Given a unital linear functional $L$ acting on state polynomials, consider a state Hankel matrix $\Gamma(L)$ indexed by
$\av{A_0^\dag} A_0, \av{A_1^\dag} A_1,
\dots, \av{A_n^\dag} A_n$, where we set $A_0 = \one$.
Its entries read
\begin{align*}\label{eq:relax}
	\Gamma_{ij}(L)= L( \av{A_i^\dag} \av{A_j} \av{A_i A_j^\dag}).
\end{align*}
Note that $\Gamma(L)_{00} = 1$, $\Gamma(L)_{i0} = \Gamma(L)_{ii}$ $\forall i$, and
$\Gamma(L) \succeq 0$.
As a consequence,
\begin{equation*}
 \max_\Gamma \quad \sum_{i=1}^n \Gamma_{ii} \quad \text{s. t.} \quad
 \Gamma_{00}=1, \quad
 \Gamma_{i0} = \Gamma_{ii} \ \forall i, \quad
 \Gamma \succeq 0
\end{equation*}
bounds $\beta$ from above.
If $\Gamma$ achieves the maximum in the above optimization problem, 
then so does $(\Gamma + \Gamma^T)/2$.
We thus can additionally impose that $\Gamma_{ij} = 0$ if $A_i A_j = - A_j A_i$
(or if $i\sim j$ in terms of $G$).
The resulting optimization problem is then nothing else than the Lov\'asz theta number, 
with the corresponding bound
$\sum_{i=1}^n \langle A_i \rangle^2 \leq \vartheta(G)$ given in \cite{de2022uncertainty}.
The idea of ~\cite{Mor_n_2024} is that this efficiently computable bound on $\vartheta$
can be strengthened
by considering a hierarchy of state Hankel matrices: indexing by products in $\av{A_i^\dag} A_i, i=0,\dots,n,$ with degrees at most $2d$ leads to the bound $\beta \leq \vartheta_d$ (see Table~\ref{tab:uncert7}). 
More generally, there is a complete SDP hierarchy based on state polynomial optimization to obtain quantum uncertainty relations.

The odd-hole inequalities strengthen the SDP formulation of the Lov\'asz number further:
let $H \subseteq G$ be a subset of vertices inducing an odd cycle $C_{2k+1}$.
Then~\footnote{Note that the original argument in \cite{Mor_n_2024} is incomplete.}
\begin{equation*}
  \sum_{i \in H} \langle A_i \rangle^2 \leq
  \left\lfloor \frac{|H|}{2} \right\rfloor.
 \end{equation*}
This can be seen by \cite[Lemma 39]{PRXQuantum.5.020318} 
which shows that $\beta(C_{2k+3}) - \beta(C_{2k+1}) \leq 1$, 
where $C_\ell$ is the cycle graph with $\ell$ vertices.
Noting that $\beta({C_3}) = 1$ yields the claim.

Table~\ref{tab:uncert7} shows bounds on $\beta$
for selected graphs of seven vertices.
There are $43$ non-isomorphic graphs with seven vertices
for which $\vartheta < \beta$.
Combining the Lov\'asz number with the odd-hole inequalities,
one obtains tight bounds on $18$ of these.
For all but $10$ graphs one has $\beta = \vartheta_2$;
for all but $7$ graphs one has  $\beta = \vartheta_3$.
This shows the numerical effectiveness of these relaxations.

 \begin{table}[tbp]
\includegraphics[width=1\textwidth]{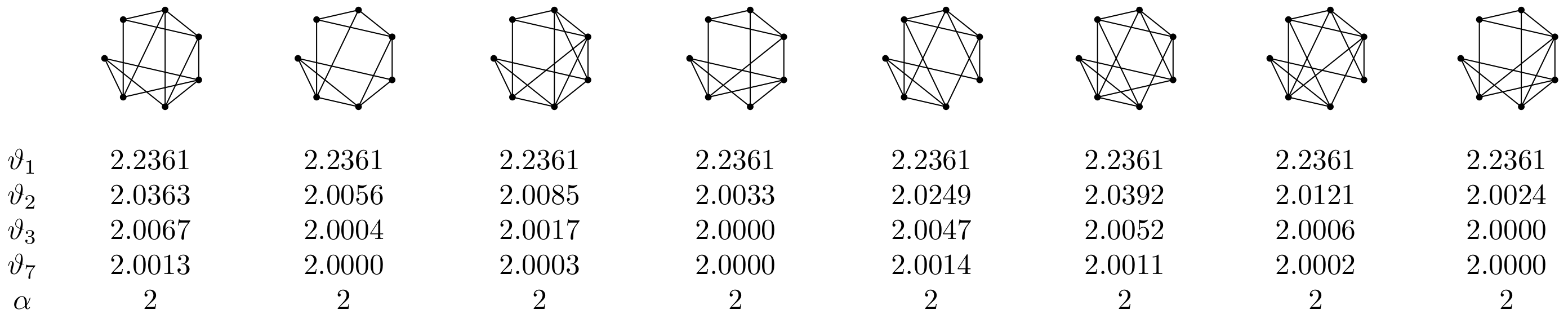}
\caption{
Upper bounds on $\beta$ for selected graphs of seven vertices. The lower bound is given by the independence number $\alpha$. \label{tab:uncert7}}
\end{table}

Finally, note the relation with the state polynomial optimization hierarchy:
it is clear that $\Gamma$ is a 
principal submatrix of the state Hankel matrix
at the second level of the state polynomial hierarchy
that is indexed by all state monomials of degree at most two.
Indexing the Hankel matrix with all state monomials of degree at most  $d$
as explained above, 
one obtains a sequence of semidefinite upper bounds that converges to $\beta$ as $d \to \infty$.
Such hierarchy can easily be adapted to operators with other defining relations,
for example to Heisenberg-Weyl operators.

\subsection{Bounds on quantum codes}

Quantum error-correcting codes
protect quantum information from noise and are thought to play an indispensable role in quantum computing devices.
The idea is to encode a state $\varrho \in \cB(\C^K)$ into a subspace $\mathcal{Q}$ of an $n$-qubit system,
so that for a given noisy channel $\mathcal{N}(\varrho) = \sum_{\mu} E_\mu \varrho E_\mu^\dag$ with $\sum_\umu E_\mu^\dag E_\mu = \mathbb{I}$, 
there exists a decoding map $\mathcal{D}$
satisfying
$\mathcal{D} \circ \mathcal{N} (\varrho)
= \varrho$ for all $\varrho$ on $\mathcal{Q}$.

Consider the case where each Kraus operator $E_\mu$ of the noisy channel acts  only on few qubits non-trivially. 
For this, let $\mathcal{E}_n$ be the $n$-qubit Pauli basis, 
formed by $n$-fold tensor products of the identity matrix $\one$ and the three Pauli matrices
\begin{align*}
    X &= \begin{pmatrix}
          0 & \phantom{-}1 \\ 1 & \phantom{-}0
        \end{pmatrix},
    &Y &= \begin{pmatrix}
          0 & -i \\ i & \phantom{-}0
        \end{pmatrix},
    &Z &= \begin{pmatrix}
          1 & \phantom{-}0 \\ 0 & -1
        \end{pmatrix}.
\end{align*}
Denote by $\operatorname{wt}(E)$ the number of subsystems (coordinates) an operator $E \in \mathcal{E}_n$ acts on non-trivially. 
For example, $\operatorname{wt}(X \otimes \one \otimes Z) = 2$.

Now a quantum code can be entirely characterized by the subspace $\mathcal{Q}$ into which the quantum information is encoded.
The Knill-Laflamme conditions provide necessary and sufficient conditions for a $\mathcal{Q}$ to be a quantum error correcting code~\cite{PhysRevA.55.900}:
a subspace defined by a hermitian projection $\Pi \in  \cB((\C^2)^{\otimes n})$ of rank $K$ is a quantum code of distance $\delta$,
if and only if
\begin{equation*}\label{eq:KLF}
    \Pi E_\mu E_\nu \Pi = c_{\mu\nu} \Pi,
\end{equation*}
for all $E_{\mu}, E_{\nu}\in \mathcal{E}_n$
with $\operatorname{wt}(E_\mu^\dag E_\nu) < \delta$.
A code with distance $\delta$ then is able to correct
all errors acting upon at most
$\lfloor \frac{n-1}{2} \rfloor$ subsystems (and linear combinations thereof).
A code is \emph{pure} if $c_{\mu\nu} = 0$ for all $E_{\mu}, E_{\nu}\in \mathcal{E}_n$ with
$0<\wt(E_\mu E_\nu) <\delta$.
A code is that is pure and has size $K=1$
is termed \emph{self-dual}.

A fundamental problem in quantum coding theory is to determine for what parameters
$(\!(n,K,\delta)\!)$ a quantum code exists.
A recent work shows that this can be formulated as a state polynomial optimization problem~\cite{munne2024sdpboundsquantumcodes}.
The key insight is that,
given a Hermitian projection $\Pi$ on $(\C^2)^{\ot n}$,
the Knill-Laflamme conditions~\eqref{eq:KLF} can be formulated as the condition
\begin{equation}\label{eq:KLF_enum}
    K B_j =
    A_j
    \qquad \text{for all}\quad 0<j<\delta,
\end{equation}
where the weight enumerators $A_j,B_j$ are given by
\begin{align*}
    B_j &= \sum_{E \in \mathcal{E}_n, \wt(E) = j} \tr(E \Pi E^\dag \Pi) , \\
    A_j &= \sum_{E \in \mathcal{E}_n, \wt(E) = j} \tr(E \Pi) \tr(E^\dag \Pi).
\end{align*}

Clearly, the RHS of \eqref{eq:KLF_enum} is a state polynomial in $\varrho = \Pi/K$
while the LHS is not.
However, the quantum MacWilliams identity allows to recover the $B_j$ from the~$A_j$ through a polynomial transform~\cite{
PhysRevLett.78.1600}:
\begin{equation*}
B(x,y) = A\left(\frac{x + 3y}{2}, \frac{x-y}{2}\right),
\end{equation*}
with
$A(x,y) = \sum_{j=0}^n A_j x^{n-j} y^j $ and
$B(x,y) = \sum_{j=0}^n B_j x^{n-j} y^j $.

Thus a strategy to formulate a state polynomial optimization hierarchy for the existence of quantum codes
with given parameters $
(\!(n,K,\delta)\!)$ emerges:
optimize over state Hankel matrices indexed
by state words in elements of the Pauli basis~$\mathcal{E}_n$.
At level $d$,
a linear combination of Hankel matrix entries
yields approximate enumerators $A^{(d)}$
and, with the help of the quantum MacWilliams transform,
also the approximate enumerators $B^{(d)}$.
In the limit $d\to\infty$ both $A_j$ and $B_j$ converge to enumerators from some state $\varrho$.
As a consequence
one can impose the Knill-Laflamme conditions~\eqref{eq:KLF_enum} directly onto the hierarchy.

It remains to see that such hierarchy can be made complete:
one needs to enforce that the state $\varrho$ in the state polynomial is proportional to a projection $\Pi$ of rank $K$ (so to correspond to a subspace $\mathcal{Q}$);
and furthermore that $\Pi$ acts on an $n$-qubit system.
The constraint that $\varrho = \Pi/K$ with $\Pi$ a projection can be imposed by the fact that the expression
$\tr(\varrho^\ell) = \tr( \eta(1, \dots, \ell) \varrho^{\otimes \ell}) = 1/K^{\ell-1}$
for all $\ell \in \N$ can be expanded in terms of state polynomials with letters in $\mathcal{E}_n$.
Here $(1,\dots, \ell) \in S_\ell$ is a cyclic permutation, and $\eta$ acts on $((\C^2)^{\ot n})^{\ot \ell}$ by permuting its $\ell$ tensor factors. 
The constraint that $\Pi$ acts on an $n$-qubit system can be recovered through a characterization
of quasi-Clifford algebras by Gastineau-Hills~\cite{gastineau-hills_1982}: 
every algebra with $m$ generators satisfying
$\alpha_i \alpha_j = (-1)^{\xi_{ij}} \alpha_j \alpha_i$,
$\xi_{ij} \in \{0,1\}$
and $\alpha_i^2=1$ is isomorphic to a direct sum of
$s$-qubit systems with $m=r+2s$ where $r$ is the number of summands. Imposing the relations of the Pauli basis ($XY = iZ$, etc) onto the state polynomial hierarchy then recovers the condition that $\Pi$ acts on the $n$-qubit space~$(\C^2)^{\ot n}$.

The hierarchy indexed by the Pauli basis $\mathcal{E}_n$
is too large to be practical at the first level. It can nevertheless be made useful:
following the strategy by~\cite{GIJSWIJT20061719}, 
an averaging over distance- and identity-preserving automorphisms 
of the graph formed by the Pauli basis
and a subsequent symmetry-reduction using the quaternary Terwilliger algebra excludes quantum codes with parameters 
$(\!(7,1,4)\!)$~\footnote{While the non-existence of the $(\!(7,1,4)\!)$ code was previously achieved through a analytical method~\cite{PhysRevLett.118.200502}, this approach provides a numerical infeasibility certificate.}, $(\!(8,9,3)\!)_2$, and $(\!(10,5,4)\!)_2$~\cite{munne2024sdpboundsquantumcodes}.

\subsubsection*{Quantum Lov\'asz and Delsarte bounds}
From the hierarchy sketched in the previous section, quantum Lov\'asz and Delsarte bounds on the existence of quantum codes can be derived.
For this we consider the same construction as done in Section~\ref{sec:uncert}, but for a feasibility instead of maximization problem.
Given a unital linear functional $L$ acting on state polynomials, consider a state Hankel matrix $\Gamma(L)$ indexed by
$\av{E_i^\dag} E_i$ with $E_i \in \mathcal{E}_n$,  where we set $E_0 = \one$.
Its entries read
\begin{align*}\label{eq:relax}
	\Gamma_{ij}(L)= L( \av{A_i^\dag} \av{A_j} \av{A_i A_j^\dag}).
\end{align*}
Note that $\Gamma(L)_{00} = 1$, $\Gamma(L)_{i0} = \Gamma(L)_{ii}$ $\forall i$, and
$\Gamma(L) \succeq 0$.
If $\Gamma(L)$ has this structure, 
then so does $(\Gamma(L) + \Gamma(L)^T)/2$.
We thus can additionally impose that $\Gamma(L)_{ij} = 0$ if $E_i E_j = - E_j E_i$.

The interesting aspect of this construction is the fact that the set of matrices $\Gamma$ with this structure corresponds to the {\em Lov\'asz theta body} 
for a graph $G$ with vertex set $\mathcal{E}_n \setminus \one$, 
where $i \sim j$ if $E_i E_j = - E_j E_i$~\footnote{Compare this with the Lov\'asz theta number of Section~\ref{sec:uncert}, which is the maximization of the sum of $\mathrm{TH}(G)$.}:
\begin{equation*}
 \mathrm{TH}(G) = \Big\{ \operatorname{diag}(X)
 \,\big|\, \begin{pmatrix}
            1 &x^T \\
			x &X
           \end{pmatrix}  \succeq 0,\,
           x_i = X_{ii} \,\forall i,\, X_{ij} = 0 \text{ if } i \sim j \Big\} .
\end{equation*}

In the case of self-dual quantum codes ($K=1$ and pure),
one additionally needs to impose the condition
$\Pi E_i^\dag E_j \Pi = 0$ 
if $0< \wt(E_i^\dag E_j) <\delta$.
Thus for a self-dual code of block-length $n$ and distance $\delta$
define its confusability graph with vertex set $\mathcal{E}_n \setminus \one$ as:
\begin{align*}
   & i\sim j \qquad \text{if} \quad
   E_i E_j = -E_j E_i \quad \text{or} \quad
   0<\wt(E_i^\dag E_j) <\delta ,  \\
   & i\sim i \qquad\, \text{if}\quad 0<\wt(E_i) < \delta.
\end{align*}
Note the appearance of loops in the graph,
arising from the case when $E_j = \one$.
We now impose normalization of the state ($K=1$), which leads to the condition 
$\sum_i \Gamma_{ii} = 2^n$.
As a consequence, the existence of a self-dual quantum code with parameters
$(\!(n,1,\delta)\!)$
can be bounded by a feasibility problem over the theta body $\mathrm{TH}(G)$.
A further relaxation of this feasibility program is the necessary condition $2^n \leq \vartheta(G) +1$ for a code to exist,
where $\vartheta$ is the Lov\'asz theta number introduced in Section~\ref{sec:uncert}.
This bound already excludes the existence of a $(\!(4,1,3)\!)$ quantum code~\footnote{The non-existence was previously known through the so-called shadow inequalities~\cite{PhysRevA.69.052330}.}.

Similarly to the classical case, one obtains a quantum Delsarte bound by averaging over distance-preserving 
anti-commutativity graph 
automorphisms of a minor variation of this quantum Lov\'asz bound. It states that if a quantum code with parameters $(\!(n,1,\delta)\!)$ exists, then the following set is non-empty:
\begin{align*}
	\Big\{
 a_j
 \,\Big|\,
 &a_0 = 1,
a_j =0 \text{\, for \,} 1<j<\delta ,\,
 \sum^n_{k=0} a_j = 2^n,\\
&\quad \quad a_j \geq 0 \text{\, and \,}
\sum^n_{i=0} K_j(i)a_i \geq 0
 \text{\, for \,}0\leq j\leq n
\Big\},
\end{align*}
where $K_j$ denotes the quarternary Krawtchouk polynomial defined by 
$K_j(i) := \sum_{\alpha=0}^i (-1)^\alpha 3^{j-\alpha}
		\binom{i}{\alpha}  \binom{n-i}{j-\alpha}$~\cite{MacWilliams1981}.

\begin{acknowledgement}
FH was supported by the Agence National de Recherche grant ANR-23-CPJ1-0012-01 and the Region Nouvelle-Aquitaine grant 34982420,
and the work was written while FH was affiliated with the Bordeaux Computer Science Laboratory (LaBRI), University of Bordeaux, France.
VM was supported by the HORIZON–MSCA-2023-DN-JD of the European Commission under the Grant Agreement No 101120296 and the project COMPUTE, funded within the QuantERA II Programme that has received funding from the EU's H2020 research and innovation programme under the GA No 101017733.
JV was supported by the National Science Foundation grant DMS-2348720.
\end{acknowledgement}
\ethics{Competing Interests}{
The authors have no conflicts of interest to declare that are relevant to the content of this chapter.}


\end{document}